\documentclass[12pt]{amsart}

\usepackage[vcentermath]{youngtab}
\usepackage{pgf}
\usepackage{graphicx}
\usepackage{amsmath, amsthm, amssymb}

\usepackage{young}

\newcommand\Z{{\mathbb Z}}
\newcommand\N{{\mathbb N}}

\newcommand\C{{\mathbb C}}

\newcommand\be{{\bf 1}}

\newcommand\Cc{{\mathcal C}}
\newcommand\Dd{{\mathcal D}}
\newcommand\Tt{{\mathcal T}}
\newcommand\Ff{{\mathcal F}}
\newcommand\Oo{{\mathcal O}}
\newcommand\Aa{{\mathcal A}}
\newcommand\Zz{{\mathcal Z}}

\newcommand\Bb{{\mathcal B}}
\newcommand\ssl{\mathfrak{sl}}
\newcommand\asl{\hat{\mathfrak{sl}}}
\newcommand\aslnm{(\hat{\mathfrak{sl}}_n)_m}
\newcommand\aslmn{(\hat{\mathfrak{sl}}_m)_n}

\newcommand\Ccaslnm{\mathcal{C}(\hat{\mathfrak{sl}}_n)_m}
\newcommand\Ccaslmn{\mathcal{C}(\hat{\mathfrak{sl}}_m)_n}

\DeclareMathOperator{\Hom}{Hom}

\DeclareMathOperator{\FP}{FPdim}
\DeclareMathOperator{\SL}{SL}
\DeclareMathOperator{\Tr}{Tr}

\newcommand\aLam{\widehat{\Lambda}}
\newcommand\alam{\widehat{\lambda}}
\newcommand\flam{\underline{\lambda}}
\newcommand\fLam{\underline{\Lambda}}
\newcommand\alamt{\widehat{\lambda^t}}

\newtheorem{thm}{Theorem}[section]
\newtheorem{prop}[thm]{Proposition}
\newtheorem{lem}[thm]{Lemma}
\newtheorem{cor}[thm]{Corollary}

\theoremstyle{remark}
\newtheorem{rem}[thm]{Remark}

\theoremstyle{definition}

\newtheorem{eg}[thm]{Example}
\address{Department of Mathematics, University of Oregon, Eugene, OR 97403-1222, USA}
\author{Victor Ostrik}
\author{Michael Sun}
\email{vostrik@uoregon.edu, msun@uoregon.edu}
\dedicatory{Dedicated to Igor Frenkel on the occasion of his 60th birthday}
\title{Level-rank duality via tensor categories}
\begin{document}
\begin{abstract} We give a new way to derive branching rules for the conformal embedding
$$(\asl_n)_m\oplus(\asl_m)_n\subset(\asl_{nm})_1.$$ In addition, we show that
the category $\Cc(\asl_n)_m^0$ of degree zero integrable highest weight $(\asl_n)_m$-representations is braided equivalent to $\Cc(\asl_m)_n^0$ with the reversed braiding. 
\end{abstract}

\maketitle
%\doublespace
%\tableofcontents

\section{Introduction}

Natural multiplicity-free representations and their associated dualities, such as Schur-Weyl duality and Howe duality, play a prominent role in the classical representation theory. In 1981, I.~Frenkel \cite{Fr} discovered an example of such a duality in the setting of affine Lie algebras. This duality relates integrable highest weight representations of affine Lie algebras $\asl_m$ at level $n$ and $\asl_n$ at
level $m$ and is known as {\em level-rank duality}. Since \cite{Fr}, this duality and its implications for conformal field theory, as well as algebraic geometry, has attracted considerable interest. See for example \cite{ABI,B,MO,NT,X1}.

The first goal of this paper is to give a simple proof of a statement that is central to level-rank duality. This statement says that the restriction of
a simple integrable highest weight representation of $\asl_{mn}$ at level $1$ to $\asl_{m}\oplus\asl_n$ is
multiplicity-free, and moreover, includes a precise description of its direct summands. This is presented as Theorem \ref{brules} in this paper. Our proof is based on ideas from \cite{W} and thus we show that the somewhat physical arguments from \cite{W} make perfect sense in the framework of tensor categories. In addition, we provide an argument which completes the proof in \cite{W}. The main ingredients
of the proof are the classical skew Cauchy formula (see Proposition \ref{scf}), 
the formula of R.~Stanley for
quantum dimensions of $\asl_m$-modules (see Proposition \ref{stf}), 
the formula of J.~Fuchs for tensoring by invertible objects in a category of integrable highest weight
$\asl_n-$modules (see Proposition \ref{rotation})
and the theory of commutative algebras in braided
tensor categories (see Section \ref{etce}). 

The second goal of this paper is to make precise the relationship between tensor categories associated
with $\asl_m$ at level $n$ and $\asl_n$ at level $m$. Our main result in this direction is Theorem
\ref{equi} which states that certain large tensor subcategories of the categories in question are braided equivalent (as long as one of the braidings is reversed). This result seems
to be well-known to the experts but we did not find it in the literature, nonetheless \cite{B} features a closely related result. Notice that Theorem \ref{equi} and the more elaborate result discussed
in Remark \ref{dyseq} imply the usual identities between fusion coefficients and entries 
of the $S-$matrix, for example \cite[(2.28)]{ABI} and \cite[(1.2)]{NT}.
As another application of Theorem \ref{equi} we give a purely algebraic version of 
some results of Feng Xu from \cite{X2} based on his theory
of {\em mirror extensions}.

We hope that the approach developed in this paper will be useful for other instances of level-rank
duality, such as those relating the affine symplectic and orthogonal Lie algebras, as well as for the general study of branching rules under conformal embeddings. 

It is with great pleasure that we dedicate this paper to Igor Frenkel who discovered the phenomena discussed
here. The first named author also thanks Peter Tingley for interesting discussions, as well as Alexey Davydov and Feng Xu for their useful comments. Both authors wish to thank the referees for their helpful suggestions.

\section{Preliminaries}

%\subsection{Notation} Symbols which have global meaning will be noted in the margin next to where they first appear.

%%%%%%%%%%%%%%
\subsection{Combinatorics}
Let $\lambda =(\lambda_1\ge \ldots \lambda_k>0=\lambda_{k+1}=\ldots)$ be a \emph{partition} of $|\lambda|:=\lambda_1+\ldots +\lambda_k$ (see e.g. \cite{Mac}). We define $h(\lambda):=k$ and will identify $\lambda$ with its corresponding \emph{Young diagram}, thus $h(\lambda)$ is just the number of rows in this diagram. We write $I_n$ for the set of all partitions with $h(\lambda)\le n$.
Let $I_{n,m}$ be the set of all $\lambda \in I_n$ with $\lambda_1 \le m$. Hence $\lambda \in I_{n,m}$ if and only if its Young diagram
fits into an $m\times n$ rectangle. Denote by $\lambda^t$ the \emph{transposed} partition of $\lambda$. Clearly, $\lambda\in I_{n,m}$ implies $\lambda^t\in I_{m,n}$.

Let $C_n:=\N^{n-1}$ be the set of \emph{dominant $\ssl_n-$weights} and let $C_{n,m}:=\{ (a_0,a_1,\ldots a_{n-1})\in \N^{n}|\sum_ia_i=m\}$ be the set of \emph{dominant $\asl_n-$weights of level $m$}.
For any $\lambda \in I_n$ we define $w_n(\lambda):=(\lambda_1-\lambda_2,\lambda_2-\lambda_3,\ldots,\lambda_{n-1}-\lambda_n)\in C_n$. 
For any $\lambda \in I_{n,m}$ we define $w_{n,m}(\lambda):=(m-\lambda_1+\lambda_n,
\lambda_1-\lambda_2,\lambda_2-\lambda_3,\ldots,\lambda_{n-1}-\lambda_n)\in C_{n,m}$.  The maps $w_n$ and $w_{n,m}$ are surjective but not injective. The map $d_{n,m}: C_{n,m}\to
I_{n,m}$ sending $(a_0,a_1,\ldots a_{n-1})$ to partition $(a_1+\ldots +a_{n-1}, a_2+\ldots +a_{n-1},\ldots,a_{n-1},0,\ldots)$
is a one-sided inverse of $w_{n,m}$. Thus $d_{n,m}$ is a bijection between $C_{n,m}$ and $I_{n-1,m}$.

Observe that
\begin{equation}\label{dmn}
|d_{n,m}(a_0,\ldots,a_{n-1})|=\sum_iia_i.
\end{equation}
For $a=(a_0,\ldots,a_{n-1})\in C_{n,m}$, we say that $|d_{n,m}(a)|+n\Z\in\Z/n\Z$ is the \emph{degree} of $a$ and write it $\deg(a)$. 
Let $\rho_n: C_{n,m}\to C_{n,m}$ be the cyclic permutation $\rho_n(a_0,a_1,\dots,a_{n-1})=(a_{n-1},a_0,\dots,a_{n-2})$.
It is clear from equation \eqref{dmn} that 
\begin{equation}\label{degrot}
\deg(\rho_n(a))=\deg(a)+m.
\end{equation}

For any $i\in \Z$ let $C_{n,m}^i=\{ a\in C_{n,m}|\deg(a)=i\pmod{n}\}$. We define
$\tau_i^{n,m}: C_{n,m}^i\to C_{m,n}^i$ by 
\begin{equation}\label{dtau}
\tau_i^{n,m}(a)=\rho_m^{\frac{i-|d_{n,m}(a)|}n}(w_{m,n}(d_{n,m}(a)^t)).
\end{equation}
It follows from equation \eqref{degrot} and the definition of $w_{m,n}$ that indeed $\tau^{n,m}_i(a)\in C_{m,n}^i$. 
We set $\tau^{n,m}:=\tau_0^{n,m}$.
Notice that the sets $C_{n,m}^i, C_{m,n}^i$ and the map $\tau^{n,m}_i$ depend only on $i$ modulo $mn$. In what follows we will often write $a^{\tau_i}$ and $a^\tau$ for $\tau_i^{n,m}(a)$ and 
$\tau_0^{n,m}(a)$ when the values of $n$ and $m$ are clear from the context.

%\hline

\subsection{Lie algebras}
\subsubsection{Finite dimensional Lie algebras}
Let $\ssl_n$ be the Lie algebra of traceless $n\times n$ matrices over $\C$. 
It is well-known that simple finite dimensional $\ssl_n-$modules are parametrized by
their highest weights, that is by the set $C_n$ (see e.g \cite{Hum}). For $\lambda \in I_n$ we will denote by $\flam$ the simple $\ssl_n-$module
with highest weight $w_n(\lambda)$.  We set $\Lambda_0=0\in C_{mn}$; also for $i=1,\ldots, mn-1$ let $\Lambda_i=(\delta_{ij})_{1\leq j< nm}\in C_{mn}$, 
which are weights of $\ssl_{nm}$ with corresponding representations $\fLam_i$. Thus $\fLam_0$ is the trivial representation
of $\ssl_{nm}$ and $\fLam_1=\C^{nm}$ is the natural representation.

The identification $\C^n\otimes\C^m=\C^{nm}$ defines an embedding $\ssl_n\oplus\ssl_m\subset\ssl_{nm}$.
Let $\Tr_n$ be the \emph{trace form} $\Tr_n(X,Y)=\Tr(XY)$ on $\ssl_n$.

\begin{lem}\label{trace form}
$\Tr_{nm}|_{\ssl_n\oplus\ssl_m}=m\Tr_n\oplus n\Tr_m$.\qed
\end{lem}

Since $\fLam_i$ is isomorphic to the exterior power $\wedge^i(\C^{nm})$, the following result is just the
classical Skew Cauchy Formula, see e.g. \cite[8.4.1]{P}.

\begin{prop}\label{scf}The restriction of $\fLam_i$ to $\ssl_n\oplus\ssl_m$ is isomorphic to
$$\bigoplus_{\lambda \in I_{n,m}, |\lambda|=i}\flam\otimes \flam^t.$$\qed
\end{prop}

\subsubsection{Affine Lie algebras}
Let $\asl_n$ be the affine Lie algebra corresponding to $\ssl_n$ (see e.g. \cite[Chapter 7]{Kac} or \cite[Section 7.1]{BK}). Thus
$\asl_n=\ssl_n\otimes \C[t,t^{-1}]\oplus \C K$ where the element $K$ is central and
$$[x\otimes t^a,y\otimes t^b]=[x,y]\otimes t^{a+b}+aTr_n(x,y)\delta_{a,-b}K.$$
An $\asl_n-$module is of level $k\in \C$ if $K$ acts as multiplication by $k$.
 %\marginpar{$\asl_n$}

For an integer $m>0$, let $\Cc\aslnm$ denote the category whose objects are finite direct sums of
integrable highest weight $\asl_n$-modules of level $m$ (see \cite[Chapter 10]{Kac} or \cite[Section 7.1]{BK}). The category $\Cc\aslnm$ is semisimple and the isomorphism classes of simple objects in $\Cc\aslnm$ are parametrised by $C_{n,m}$ via their highest weights, see {\em loc. cit}.
Given $\lambda \in I_{n,m}$ or $a\in C_{n,m}$ we will write $\alam$ or $\hat a$ for a simple object in $\Cc\aslnm$ with highest weight $w_{n,m}(\lambda)$ or $a$. We will write $\aLam_0,\dots,\aLam_{nm-1}$ for the simple objects in $\Cc(\asl_{nm})_1$.

For $i\in \Z/n\Z$ let $\Cc\aslnm^i$ be the full subcategory of $\Cc\aslnm$ consisting of direct sums of simple modules with highest weights of degree $i$. Thus we have a decomposition
\begin{equation}\label{graded}
\Cc\aslnm =\bigoplus_{i\in \Z/n\Z}\Cc\aslnm^i.
\end{equation}

The embedding $\ssl_n\oplus\ssl_m\subset\ssl_{nm}$ induces an embedding $\asl_n\oplus\asl_m\subset\asl_{nm}$.
Lemma \ref{trace form} shows that pulling back $(\asl_{nm})_k$-modules under this embedding gives $\asl_n-$modules of level $mk$ and
$\asl_m-$modules of level $nk$. We express this symbolically as: 
\begin{equation}\label{embedk}(\asl_n)_{mk}\oplus(\asl_m)_{nk}\subset(\asl_{nm})_k.\end{equation}

It is known (see e.g. \cite[Section 12.10]{Kac}) that the restrictions of integrable highest weight $\asl_{nm}-$modules decompose as (possibly infinite) direct sums of integrable highest weight $\asl_n\oplus\asl_m-$modules. An irreducible $\asl_n\oplus\asl_m-$module decomposes as a tensor product of an irreducible $\asl_n-$module and an irreducible $\asl_m-$module.

Let $\asl_n^+\subset \asl_n$ be the subalgebra spanned by $\{x\otimes t^a\,: x\in \ssl_n, a>0\}$. For $V\in \Cc\aslnm$ we set
$$(V)^{\asl_n^+}=\{ v\in V \,:\, xv=0\quad\text{for all }x\in \asl_n^+\} .$$ 
It is clear that $(V)^{\asl_n^+}$ is a subspace of $V$ invariant under $\ssl_n=\ssl_n\otimes 1\subset \asl_n$.
The following result is well known.

\begin{lem}\label{fplus} 
For any $\lambda \in I_{n,m}$ we have an isomorphism of $\ssl_n-$modules $(\alam)^{\asl_n^+}=\flam$.
\end{lem}

{\bf Sketch of proof.} The $\ssl_n-$module $(V)^{\asl_n^+}$ is integrable, hence it is a sum of finite dimensional modules. 
If $\flam$ is a direct summand of $(V)^{\asl_n^+}$ then we have a nonzero homomorphism $V^m_{w_n(\lambda)}\to V$
where $V^m_{w_n(\lambda)}$ is the {\em Weyl module} with highest weight $w_n(\lambda)$, see e.g. \cite[Section 7.1]{BK}. 
Since the only integrable quotient of $V^m_{w_n(\lambda)}$ is $\alam$ (or zero if $w_n(\lambda)\not \in C_{n,m}$) we deduce the result from the properties of the {\em homogeneous grading}, see \cite[Section 7.1]{BK}.\qed

\subsubsection{Conformal embeddings}
Recall that on any integrable highest weight module over an affine Lie algebra $\hat{\mathfrak{g}}$ there is an action of the Virasoro algebra 
via {\em Sugawara operators}, see e.g. \cite[Proposition 7.4.4]{BK} or \cite[Section 12.8]{Kac}; the central charge of this action is given by
\cite[7.4.5]{BK} or \cite[(12.8.10)]{Kac}. Thus on any integrable highest weight $\asl_{nm}-$module of level $k$ we have an action of the Virasoro algebra with central charge
$$c_{(\asl_{nm})_k}=\frac{(n^2m^2-1)k}{nm+k}$$
and another action arising from the restricted $(\asl_n)_{mk}\oplus(\asl_m)_{nk}-$action, which has central charge
$$c_{(\asl_{n})_{mk}\oplus (\asl_m)_{nk}}=\frac{(n^2-1)mk}{n+mk}+\frac{(m^2-1)nk}{m+nk}.$$

One of our main goals is to study the restrictions of modules under the the embedding \eqref{embedk} with $k=1$. This case is singled out because here the central charges agree with both equal to $nm-1$.
\begin{equation}\label{cequ}
 c_{(\asl_{nm})_1}=c_{(\asl_{n})_{m}\oplus (\asl_m)_{n}}.\end{equation}
%\frac{n^2m^2-1}{nm+1}=\frac{(n^2-1)m}{n+m}+\frac{(m^2-1)n}{m+n}=
Thus we say that 
\begin{equation}\label{embed}(\asl_n)_m\oplus(\asl_m)_n\subset(\asl_{nm})_1\end{equation}
is a {\em conformal embedding}. It is known that \eqref{cequ} implies that a restriction of an integrable
highest weight $\asl_{nm}-$module of level $1$ to $(\asl_n)_m\oplus(\asl_m)_n$ is a {\em finite} direct sum of simple
$(\asl_n)_m\oplus(\asl_m)_n-$modules, see \cite[13.14.6]{Kac}.
 Equivalently, the restrictions of $\aLam_i$ are finite direct sums of modules of the form
$\alam \boxtimes \widehat \mu$ with $\lambda \in I_{n,m}$ and $\mu \in I_{m,n}$, where we will use $\boxtimes$ to denote the usual tensor product of modules over $\aslnm$ and $\aslmn$ considered
as an $\aslnm\oplus \aslmn-$module to avoid confusion with the {\em fusion tensor product} $\otimes$ to appear in Section \ref{fus}. We will write $\alam \boxtimes \widehat \mu \subset \aLam_i$ if $\alam \boxtimes \widehat \mu$ appears 
in the restriction of $\aLam_i$ with nonzero multiplicity.

%The identification $\C^n\otimes\C^m=\C^{nm}$ defines an embedding $\ssl_n\oplus\ssl_m\subset\ssl_{nm}$ which in turn gives %an embedding of their affine Lie algebras $\asl_n\oplus\asl_m\subset\asl_{nm}$. When pulling back $(\asl_{nm})_k$-modules %under this embedding the modules obtained will have levels for their summands determined by $k$. These could be determined for %example by restricting the Killing form of $\ssl_{nm}$ to obtains multiples of the Killing forms for $\ssl_n$ and $\ssl_m$, where the %multiple is the level. In particular, when $k=1$, we get the levels
%\begin{equation}\label{embed}(\asl_n)_m\oplus(\asl_m)_n\subset(\asl_{nm})_1.\end{equation}

\section{Tensor categories}
We refer the reader to \cite[Chapters 1-3]{BK} for the basic notions of tensor categories. All tensor categories considered in this paper are {\em fusion categories} in a sense of \cite{ENO}, that is semisimple $\C-$linear rigid categories with finite dimensional spaces of morphisms, finitely many irreducible objects and an irreducible unit object. For a fusion category $\Cc$ there is a unique homomorphism from the Grothendieck ring of $\Cc$ to real numbers sending each isomorphism class to a nonnegative real number; the value of this homomorphism on the class represented by $X\in \Cc$ is called the {\em Frobenius-Perron dimension} of $X$ and written $\FP_\Cc(X)$, see \cite[Section 8.1]{ENO}. One defines $\FP \Cc =\sum_{X\in \Oo(\Cc)}\FP_\Cc(X)^2$ where $\Oo(\Cc)$ is the set  of isomorphism classes of simple objects in $\Cc$.  More generally for a \emph{Serre subcategory} $\Cc^\# \subset \Cc$ we set $\FP \Cc^\#=\sum_{X\in \Oo(\Cc^\#)}\FP_\Cc(X)^2$. We write $\Cc\boxtimes\Dd$ for the \emph{external tensor product} of $\C-$linear categories, whose objects are direct sums of pairs $(a,b)\in\Cc\times\Dd$, see \cite[Definition 1.1.15]{BK}.
It is clear that the dimension respects the external tensor product: $\FP (\Cc \boxtimes \Dd)=(\FP \Cc )(\FP \Dd)$.
\subsection{Braided tensor categories} \label{tcpr}
Recall (\cite[Definition 1.2.3]{BK}) that a monoidal category $(\Cc,\otimes)$ is \emph{braided} if there is a natural bifunctor isomorphism $c_{X,Y}:X\otimes Y\to Y\otimes X$ called a {\em braiding} subject to the hexagon axioms. 
For every braided tensor category $(\Cc,\otimes,c)$, there is a \emph{reversed} braiding on $\Cc$ given by $c_{X,Y}^{rev}=c_{Y,X}^{-1}$; a braided tensor category $\Cc$ endowed with the reversed braiding will be denoted $\Cc^{rev}$.
A functor $\Tt: \Cc\to \Dd$ is said to be a \emph{braided equivalence of categories} if it is a monoidal functor that preserves the braiding and is an equivalence of the underlying categories, see e.g. \cite[Definition 1.2.10]{BK}. 
A \emph{braid-reversing equivalence} $\Tt: \Cc\to \Dd$ is a braided equivalence $\Cc \to \Dd^{rev}$.
Two objects $X, Y$ in a braided tensor category $\Cc$ are said to \emph{mutually centralize} each other, in the sense of \cite{Mu}, if $c_{X,Y}c_{Y,X}=1_{X\otimes Y}$. For a fusion subcategory $\Dd$ of a braided category $\Cc$, its \emph{centralizer} $\Dd'$ is the full subcategory of $\Cc$ consisting of all objects that centralize every object of $\Dd$. For a tensor category $\Cc$, the \emph{Drinfeld center} $\Zz(\Cc)$ is defined to be the category whose objects are pairs $(X,\gamma_X)$ for $X$ an object of $\Cc$ and $\gamma_X:V\otimes X\simeq X\otimes V$ is a natural family of isomorphisms indexed by $V\in\Cc$, subject to certain compatibility conditions making it naturally equipped with the structure of a braided tensor category, see e.g \cite[Definition XIII 4.1]{Ka}. If $\Cc$ is braided, then we have natural braided and braid-reversing functors
\begin{equation}\label{intoz}i:\Cc\to\Zz(\Cc),\quad X\mapsto (X,c_{-,X})\end{equation}
\begin{equation}\label{intoz2}j:\Cc\to\Zz(\Cc),\quad X\mapsto (X,c_{X,-}^{-1}).\end{equation}
These functors are fully faithful and, if $\Cc$ is a fusion category, then their images satisfy $i(\Cc)'=j(\Cc)$.
Given a functor $\Ff: \Aa\to\Bb$ with $\Aa$ braided, we say $\Ff$ is \emph{central} if it factors through the forgetful functor $\Zz(\Bb)\to \Bb$ via a braided functor $\Aa\to\Zz(\Bb)$.
%Let $(\Cc,\otimes,\mathbf{1})$ be a monoidal tensor category (see, for example, \cite{BK} Chapter 1, \cite{MacL} Chapter VII).%\marginpar{$\mathbf{1}$}
%\begin{defn}\label{braid} $(\Cc,\otimes)$ is \emph{braided} if there is a natural collection of isomorphisms $c_{X,Y}:X\otimes Y\to %Y\otimes X$ subject to relations from the braid group. (See Definition 1.2.3 in \cite{BK}). For every braided tensor category $(\Cc,%\otimes,c)$, there is a \emph{reversed} braiding on $\Cc$ given by $c_{X,Y}^{rev}=c_{Y,X}^{-1}$. A functor $\Tt: \Cc\to \Dd$ is %said to be a \emph{braided equivalence of categories} if it is a monoidal functor that preserves the braiding and is an equivalence %of the underlying categories. (See, for example, \cite{MacL} Chapter IV and Chapter XI).
%A \emph{braid reversing equivalence} $\Tt: \Cc\to \Dd$ is a braided equivalence $(\Cc,\otimes,c)\to(\Dd,\otimes,d^{rev})$.
%\end{defn}

\subsection{Tensor category $\Cc\aslnm$}\label{fus}Let $\Cc(\hat{\mathfrak{g}})_k$ be the category consisting of finite direct sums of integrable highest weight modules of level $k$ over an (untwisted) affine Lie algebra $\hat{\mathfrak{g}}$.
It is a deep result that $\Cc(\hat{\mathfrak{g}})_k$ endowed with the so called {\em fusion tensor product} is a {\em modular tensor category}. This is denoted as $\Cc(\hat{\mathfrak{g}};k)$ in \cite[Chapter 7]{BK} where this is discussed. 
In particular, $\Cc\aslnm$ is a modular tensor category. The unit object of this category
is the {\em vacuum module} and corresponds to the empty Young diagram $\emptyset \in I_{n,m}$. The decomposition \eqref{graded} satisfies $\Cc\aslnm^i\otimes \Cc\aslnm^j\subset \Cc\aslnm^{i+j}$, which allows us to say that $\Cc\aslnm$ is
{\em graded}. (This grading is inherited from the grading of the category of finite dimensional $\ssl_n-$modules via the characters
of the center of the group $\SL_n$; see for example \cite[Proposition 7.3.8]{BK}). In particular, 
the subcategory $\Cc\aslnm^0$ is closed under tensor products making it a braided fusion category.

We get from \cite[Proposition 8.20]{ENO} that for any $i\in \Z/n\Z$
\begin{equation}\label{dimgrad}
\FP \Cc\aslnm^i=\FP \Cc\aslnm^0.
\end{equation}
\begin{equation}\label{dimgrad1}
\FP \Cc\aslnm=n\FP \Cc\aslnm^0. 
\end{equation}

We now explain how Stanley's formula \cite[Theorem 15.3]{Stanley} gives the Frobenius-Perron dimensions of simple objects in $\Cc\aslnm$.

Recall that for a box $T$ in Young diagram $\lambda$ one defines its {\em content} $c_T$ as difference of its $x$ and $y$ coordinates and {\em hooklength} $h_T$ as the number of boxes below it and to the right of it plus one, 
see for example \cite[Section I.1]{Mac}.

\begin{eg} \label{hook} Let $\lambda=(4,3,1)$. Its contents and hooklengths are 
\begin{Young}
0&1&2&3\cr
-1&0&1\cr
-2\cr
\end{Young}
and 
\begin{Young}
6&4&3&1\cr
4&2&1\cr
1\cr
\end{Young} respectively.
\end{eg}

For fixed $m$ and $n$, define the {\em quantum integers} to be: 
$$[i]:=\frac{\sin{\frac{\pi i}{m+n}}}{\sin{\frac{\pi}{m+n}}}\quad\text{for $i\in\Z$}.$$
Obviously,
\begin{equation} \label{mirror}
\quad[i]=[m+n-i].
\end{equation}

There is a general formula for dimensions of simple objects in the category $\Cc\aslnm$ defined in terms of the {\em ribbon structure} on $\Cc\aslnm$; it coincides with the Frobenius-Perron dimension because it only takes positive-real values, guaranteeing uniqueness. See for example \cite[(3.3.2)]{BK}. This formula gives the dimension of $\alam$ as a specialization of the character of $\flam$. Recall that the irreducible characters of $\ssl_n$ are
given by {\em Schur functions}. The specialization of the Schur functions mentioned above was computed by
R.~Stanley \cite[Theorem 15.3]{Stanley}. Thus combining \cite[(3.3.2)]{BK} and \cite[Theorem 15.3]{Stanley} we get
the following:

\begin{prop}\label{stf}
For any $\lambda \in I_{n,m}$ we have
\begin{equation} \label{dimf}
\FP_{\Cc\aslnm}(\alam)=\prod_{T\in \lambda}\frac{[n+c_T]}{[h_T]}.
\end{equation}\qed
\end{prop}

\begin{eg} Let $n=4$, $m\geq 4$ and $\lambda$ be as in Example \ref{hook}. 
$$\FP_{\Cc(\asl_4)_m}(\alam)=\frac{[4][5][6][7][3][4][5][2]}{[6][4][3][1][4][2][1][1]}=[7][5]^2.$$
\end{eg}

%\begin{rem} Notice that \eqref{dimf} implies that $\dim_{\Cc\aslnm}(\alam)>0$. Thus the dimensions in the category
%$\Cc\aslnm$ coincide with {\em Frobenius-Perron dimensions}, see \cite[]{ENO}. Thus in what follows we will freely use the results
%involving the Frobenius-Perron dimension. 
%\end{rem}

Using \eqref{mirror} we get the following

\begin{cor} \label{transdim}
For any $\lambda \in I_{n,m}$ we have
$$\FP_{\Cc\aslnm}(\alam)=\FP_{\Cc\aslmn}(\alam^t).$$\qed
\end{cor}

Let $\sigma_m$ be the partition $(m,0,0,\ldots)$. Then it follows easily from \eqref{dimf} that $\FP_{\Cc\aslnm}(\widehat{\sigma}_m)=1$.
Thus the object $\widehat{\sigma}_m\in \Cc\aslnm$ is {\em invertible}, that is for any simple object
$L\in \Cc\aslnm$ the object $\hat \sigma_m\otimes L$ is also simple. The following well-known result is a special case of
results in \cite{Fuchs}.

\begin{prop}\label{rotation}
 Let $L\in \Cc\aslnm$ be a simple object with highest weight $a\in C_{n,m}$. Then the highest weight of
$\hat \sigma_m\otimes L$ is $\rho_n(a)$.\qed
\end{prop}

Proposition \ref{rotation} implies the dimensions of simple objects with highest weights $\rho_n^i(a)$ all coincide. Recall that for any $i\in \Z$ and $a\in C_{n,m}^i$ we defined $a^{\tau_i}=\tau_i^{n,m}(a)$ in \eqref{dtau}.
Combining this with Corollary \ref{transdim} we have

\begin{cor} \label{dimcor}
For any $i\in \Z$ and $a \in C_{n,m}^i$ we have
$$\FP_{\Cc\aslnm}(\widehat{a})=\FP_{\Cc\aslmn}(\widehat{a^{\tau_i}}).$$\qed
\end{cor}

Observe that $w_{nm,1}(\sigma_1)=\Lambda_1$. Hence Proposition \ref{rotation} implies that
$$\aLam_i\cong \aLam_1^{\otimes i}\; \mbox{and}\; \aLam_1^{\otimes mn}=\aLam_0\text{ in }\Cc(\asl_{nm})_1.$$
Equivalently,
 \begin{equation}\label{l1f}
\aLam_i\otimes \aLam_j\cong \aLam_k\; \; \mbox{where}\; \; k\equiv i+j \pmod{mn}.
\end{equation}
In particular, all simple objects in $\Cc(\asl_{nm})_1$ are $1$-dimensional and
\begin{equation}\label{l1d}
\FP \Cc(\asl_{nm})_1=nm.
\end{equation}

\subsection{Etale algebras and conformal embeddings}\label{etce}

An {\em \'etale algebra} in a semisimple braided tensor category $\Cc$ is defined to be an object 
$A\in \Cc$ endowed with an associative commutative unital multiplication and that 
the category $\Cc_A$ of right $A-$modules is semisimple, see \cite[Definition 3.1]{KO2}. 
An \'etale algebra $A$ is called {\em connected} if the unit object appears in $A$ with multiplicity 1.
For a connected \'etale algebra $A\in \Cc$ the category $\Cc_A$ with operation $\otimes_A$ of 
tensor product over $A$ is naturally a fusion category, see e.g. \cite[Section 3.3]{KO2}. The category $\Cc_A$ contains a full tensor Serre subcategory $\Cc_A^{dys}$ of {\em dyslectic} modules, which is also naturally braided. See for example \cite[Section 3.5]{KO2}.
The free module functor $\Ff_A:\Cc\to\Cc_A$ given by $\Ff(X)=X\otimes A$ has the structure of a central functor given by $X\mapsto (X\otimes A, c_{-,X})$, see e.g \cite[Section 3.4]{KO2}.
%closed under direct sums and taking a direct summand; 
%moreover the tensor category $\Cc_A^{dys}$ 

\begin{prop}\label{ddim} {\em (\cite[Proposition 8.7]{ENO}, \cite[Lemma 3.11]{KO2}, 
\cite[Corollary 3.32]{KO2})}
If $A$ is a connected \'etale algebra in a modular tensor category $\Cc$, then 
\begin{itemize}
\item[(i)]$(\FP_\Cc A)\FP_{\Cc_A}M=\FP_\Cc M$ for any $M\in \Cc_A$.\\ 
\item[(ii)]$(\FP_\Cc A)\FP\Cc_A=\FP\Cc $.\\
\item[(iii)] $(\FP_\Cc A)^2\FP\Cc_A^{dys}=\FP\Cc $.
\end{itemize}\qed
\end{prop}
%\begin{proof}This is Corollary $3.30$ of \cite{KO2}. Note we will not distinguish between the Frobenius-Peron dimension ($\FP$ in %\cite{KO2}) and $\dim_q$ because they coincide for the types of categories considered herein, this is the definition of %\emph{pseudo-unitary}. See, for example, Theorem $7.0.1$ in \cite{BK}, which implies $\Ccaslnm$ (Definition \ref{csl}) pseudo-unitary.
%\end{proof}

A general result \cite[Theorem 5.2]{KO} states that for any conformal embedding the pullback of the vacuum module is an \'etale algebra; moreover taking pullbacks is a braided equivalence with the category of dyslectic 
modules over this algebra.
Specializing this to the conformal embedding \eqref{embed} we get

\begin{thm}\label{A} Let $A$ be the restriction of $\aLam_0$ under the embedding \eqref{embed}. Then 
$A$ is a connected \'etale algebra in $\Cc\aslnm \boxtimes \Cc\aslmn$. Moreover, the restriction functor is a braided equivalence $\Cc(\asl_{mn})_1\cong (\Cc\aslnm \boxtimes \Cc\aslmn)_A^{dys}$ with the pullbacks of the $\aLam_i$ being precisely the simple dyslectic $A$-modules.\qed
\end{thm}%\marginpar{$A$}
%\newpage
	
\section{Branching rules}

\subsection{Overview} We state here what we will prove in \ref{Walton} and \ref{cdim}.

\begin{thm}[Branching rules]\label{brules}
Let $m,n\geq2$ and $0\le i<mn$. There is an isomorphism of $\aslnm \oplus \aslmn-$modules:
$$\aLam_i\cong \bigoplus_{a\in C_{n,m}^i}\widehat{a}\boxtimes\widehat{a^{\tau_i}}.$$
In particular, we have
$$A:=\aLam_0\cong \bigoplus_{a\in C_{n,m}^0}\widehat{a}\boxtimes\widehat{a^{\tau}}\in\Ccaslnm^0\boxtimes\Ccaslmn^0.$$
\end{thm}

\begin{cor}$\tau_i^{n,m}: C_{n,m}^i\to C_{m,n}^i$ is a bijection with inverse $\tau_i^{m,n}$.\qed
\end{cor}

\begin{rem} It follows from the proof of Theorem \ref{brules}, otherwise easily checked combinatorially, that we can compute $\tau_i^{n,m}(a)$ in the following way: choose any $\lambda \in I_{n,m}$
with $w_{n,m}(\lambda)=a$ and compute
$$\tau_i^{n,m}(a)=\rho_m^{\frac{i-|\lambda|}{n}}(w_{m,n}(\lambda^t)).$$
\end{rem} 

The strategy of proof of Theorem \ref{brules} is a simple one and will be divided into two parts. In \ref{Walton} we show that for any $a\in C_{n,m}^i$ we have 
$\widehat{a}\boxtimes\widehat{a^{\tau_i}}\subset \aLam_i$ using the argument by M.~Walton \cite{W}, while in \ref{cdim} we show that these summands exhaust the dimension of the appropriate graded component of $\Ccaslnm\boxtimes\Ccaslmn$. This will complete the proof of Theorem \ref{brules}, as well as make the argument invoked by Walton rigorous.

\subsection{Using the classical branching rule.}\label{Walton}
%$\tau_i^{n,m}$

\begin{prop}\label{cprop} For any $\lambda\in I_{n,m}$, 
$\alam\boxtimes\alamt\subset\aLam_{|\lambda|}.$
\end{prop}

\begin{proof} It is clear that $\asl_n^+\oplus \asl_m^+\subset \asl_{nm}^+$, so 
$$(\aLam_i)^{\asl_n^+\oplus \asl_m^+}\supset (\aLam_i)^{\asl_{nm}^+}=\fLam_i\; \; \mbox{see Lemma \ref{fplus}}.$$
Now the result follows from Lemma \ref{fplus} and Proposition \ref{scf}.
\end{proof}

Since $\widehat{\sigma_m^t}\cong \mathbf{1}$, Proposition \ref{cprop} immediately implies that
\begin{equation}\label{sigmas}
\mathbf{1}\boxtimes\widehat{\sigma}_n\subset\aLam_n\; \; \mbox{and}\; \; 
\widehat{\sigma}_m\boxtimes\mathbf{1}\subset\aLam_m
\end{equation}

By Theorem \ref{A}, $A=\aLam_0\in \Ccaslnm\boxtimes\Ccaslmn$ is a connected \'etale algebra and $\aLam_i\in \Ccaslnm\boxtimes\Ccaslmn$ is a simple $A-$module.

\begin{cor}\label{crule}
There are isomorphisms of $A-$modules 
$$\begin{aligned}\aLam_n&\cong(\mathbf1\boxtimes\widehat{\sigma}_n)\otimes A\\
\text{and}\quad\aLam_m&\cong(\widehat{\sigma}_m\boxtimes\mathbf{1})\otimes A.\end{aligned}$$
\end{cor}

\begin{proof} We will only show the first isomorphism. Since we have
$$\Hom_A((\mathbf1\boxtimes\widehat{\sigma}_n)\otimes A, \aLam_n)
=\Hom(\mathbf1\boxtimes\widehat{\sigma}_n,\aLam_n)\neq0,$$
and $\aLam_n$ is simple, it is enough to show that $(\mathbf1\boxtimes\widehat{\sigma}_n)\otimes A$ is simple because the category of $A-$modules is semisimple. Hence we are reduced to the following computation, noting $A$ is connected.
$$\begin{aligned}\Hom_A((\mathbf1\boxtimes\widehat{\sigma}_n)\otimes A,(\mathbf1\boxtimes\widehat{\sigma}_n)\otimes A)&=\Hom(\mathbf1\boxtimes\widehat{\sigma}_n,(\mathbf1\boxtimes\widehat{\sigma}_n)\otimes A)\\
	&=\Hom(\mathbf1\boxtimes\mathbf1,A)=\C.\end{aligned}$$\end{proof}

\begin{lem}\label{tensn}
Assume that $\alam \boxtimes \widehat{\mu}\subset \aLam_j$. Then
$$\alam \boxtimes \widehat{\rho_m(\mu)}\subset \aLam_k\; \; \mbox{where}\; \; k\equiv j+n\pmod{mn}.$$
\end{lem}

\begin{proof} In view of Theorem \ref{A} we can think of $\aLam_i$ as an $A-$module in $\Cc\aslnm\boxtimes \Cc\aslmn$. Using \eqref{l1f} and Corollary \ref{crule} we get
$$\aLam_k\cong \aLam_n \otimes_A \aLam_j\cong (\mathbf1\boxtimes\widehat{\sigma}_n)\otimes A\otimes_A\aLam_j=(\mathbf1\boxtimes\widehat{\sigma}_n)\otimes \aLam_j\supset \alam \boxtimes (\widehat{\sigma}_n \otimes \widehat{\mu}).$$
Now the result follows from Proposition \ref{rotation}.
\end{proof}

\begin{prop}\label{lbound}
 For any integer $i$ with $0\le i<nm$ and $\lambda \in I_{n,m}$ with $|\lambda|\equiv i\pmod{n}$ we have
$\alam \boxtimes \widehat{\mu} \subset \aLam_i$ where $\mu =\rho_m^{\frac{i-|\lambda|}{n}}(\lambda^t)$. In particular
for any $a\in C_{n,m}^i$ we have $\widehat{a}\boxtimes \widehat{a^{\tau_i}}\subset \aLam_i$.   
\end{prop} 

\begin{proof} This follows from Proposition \ref{cprop} by applying Lemma \ref{tensn} $k$ times where $k>0$ and 
$k\equiv \frac{i-|\lambda|}{n}\pmod{mn}$.
\end{proof}

\subsection{Exhausting dimensions}\label{cdim}
Combining \eqref{dimgrad1}, \eqref{l1d}, Theorem \ref{A} and Proposition \ref{ddim} (iii) we obtain
%$$(\FP_{\Cc\aslnm \boxtimes \Cc\aslmn}A)^2=\frac{\FP \Ccaslnm\FP \Ccaslmn}{\FP \Cc(\asl_{nm})_1}.$$
%Using  we get
\begin{equation}\label{dima}
(\FP_{\Cc\aslnm \boxtimes \Cc\aslmn}A)^2=\FP \Cc\aslnm^0\FP \Cc\aslmn^0.
\end{equation}
%Moreover, we have
\begin{lem}\label{dimac}$\FP_{\Cc\aslnm \boxtimes \Cc\aslmn}A=\FP \Cc\aslnm^0=\FP \Cc\aslmn^0.$
\end{lem}

\begin{proof} By Proposition \ref{lbound}, $\oplus_{a\in C_{n,m}^0}\widehat{a}\boxtimes \widehat{a^\tau}\subset \aLam_0=A$. Hence
$$\begin{aligned}\FP_{\Cc\aslnm \boxtimes \Cc\aslmn}A&\ge \FP_{\Cc\aslnm \boxtimes \Cc\aslmn}(\oplus_{a\in C_{n,m}^0}\widehat{a}\boxtimes \widehat{a^\tau})\\
              \text{(Corollary \ref{dimcor})}\qquad &=\sum_{a\in C_{n,m}^0}(\FP_{\Cc\aslnm}\widehat{a})^2=\FP \Cc\aslnm^0.\end{aligned}$$
Interchanging the roles of $n$ and $m$ we get 
$$\FP_{\Cc\aslnm \boxtimes \Cc\aslmn}A\ge \FP \Cc\aslmn^0.$$
Combining these inequalities with \eqref{dima} gives the result.
\end{proof}

\begin{rem} Alternatively, one can reduce Lemma \ref{dimac} to an
elementary trigonometric identity using formula \cite[(3.3.9)]{BK} for $\FP \Cc\aslnm$.
\end{rem}
\noindent
{\bf Proof of Theorem \ref{brules}.} Since $\FP_{\Cc(\asl_{nm})_1}\aLam_i=1$ for any $i$, \eqref{dimgrad}, Proposition \ref{ddim}(i) and Lemma \ref{dimac} together give
$$\FP_{\Cc\aslnm \boxtimes \Cc\aslmn}\aLam_i=\FP_{\Cc\aslnm \boxtimes \Cc\aslmn}A=
\FP \Cc\aslnm^i.$$
On the other hand, $\aLam_i\supset\oplus_{a\in C_{n,m}^i}\widehat{a}\boxtimes \widehat{a^{\tau_i}}$ by  Proposition \ref{lbound}, hence%Corollary \ref{dimcor}
$$\begin{aligned}\FP_{\Cc\aslnm \boxtimes \Cc\aslmn}(\oplus_{a\in C_{n,m}^i}
\widehat{a}\boxtimes \widehat{a^{\tau_i}})&\stackrel{3.6}{=}\sum_{a\in C_{n,m}^i}(\FP_{\Cc\aslnm}\widehat{a})^2\\
																					&=\FP \Cc\aslnm^i.\qquad\qed\end{aligned}$$

%\newpage

\subsection{Examples}

\begin{eg} Let $n=3$ and $m=6$. Let $i=13$ and $\lambda=\tiny\yng(3,1)$. We determine which partition pairs with $\lambda$ to appear as a summand in $\aLam_{13}$ using \eqref{dtau}.
{\tiny$$\yng(3,1)\stackrel{t}{\rightarrow}\yng(2,1,1)\stackrel{w_{6,3}}{\rightarrow}(1,1,0,1,0,0)\stackrel{\rho_6^{3}}{\rightarrow}(1,0,0,1,1,0)\stackrel{d_{6,3}}{\rightarrow}\yng(2,2,2,1)$$}
$$\text{Hence }\quad{\tiny\yng(3,1)\boxtimes\yng(2,2,2,1)\subset\aLam_{13}}.$$
Note, for example, $\lambda$ does not appear as the left factor in any summand of $\aLam_9$ because $9\neq4\pmod3$.
\end{eg}
%\begin{eg} Let $n=3$ and $m=6$. Let $i=9$, $\lambda=\tiny\yng(2,1)$. Then
%{\tiny$$\yng(2,1)\boxtimes\yng(3,3,2,1)\subset\aLam_{9}.$$}
%\end{eg}

\begin{eg}Let $n=3$, $m=6$ and consider $A=\aLam_0$. We have the following decomposition:
$$\begin{aligned}A&={\tiny\mathbf1\boxtimes\mathbf1+\yng(2,1)\boxtimes\yng(2,1,1,1,1)+\yng(5,4)\boxtimes\yng(3,2,1)+
\yng(4,2)\boxtimes\yng(2,2,1,1)}\\
&{\tiny +\yng(3)\boxtimes\yng(3,3,2,2,2)
+\yng(6,3)\boxtimes\yng(2,2,2)+\yng(5,1)\boxtimes\yng(3,3,3,2,1)}\\
&{\tiny+\yng(6)\boxtimes\yng(3,3,3,3)
+\yng(3,3)\boxtimes\yng(3,1,1,1)+\yng(6,6)\boxtimes\yng(3,3)}.
\end{aligned}$$
%{\tiny$$\yng(5,4)\stackrel{t}{\rightarrow}\yng(2,2,2,2,1)=(1,0,0,0,1,1)\stackrel{\sigma^{-3}}{\rightarrow}(0,1,1,1,0,0)=\yng(3,2,1).$$}

%{\tiny$$\yng(6,3)\stackrel{t}{\rightarrow}\yng(2,2,2,1,1,1)=\yng(1,1,1)=(2,0,0,1,0,0)\stackrel{\sigma^{-3}}{\rightarrow}(1,0,0,2,0,0)=\yng(2,2,2).$$}

%{\tiny$$\begin{aligned}A&=\mathbf1\otimes\mathbf1+\yng(2,1)\otimes\yng(2,1,1,1,1)+\young(\bullet\bullet\bullet\bullet\bullet,%\bullet\bullet\bullet\bullet)\otimes\yng(3,2,1)+\yng(4,2)\otimes\yng(2,2,1,1)+\yng(3)\otimes\yng(3,3,2,2,2)\\
%&+\yng(6,3)\otimes\young(\bullet\bullet,\bullet\bullet,%\bullet\bullet)+\yng(5,1)\otimes\yng(2,2,2,1,1)+\yng(6)\otimes\yng(3,3,3,3)\\&+\yng(3,3)\otimes\yng(3,1,1,1)+\yng(6,6)\otimes\yng(3,3)
%\end{aligned}$$}
\end{eg}

\section{A tensor equivalence}
\subsection{Main theorem} The number of isomorphism classes of simple objects in the category $\Cc\aslnm$ is the cardinality of $C_{n,m}$, that is ${n+m-1\choose n-1}$.  Thus the categories $\Cc\aslnm$ and $\Cc\aslmn$ are not equivalent even as additive categories
when $n\ne m$. However we have the following

\begin{thm}\label{equi}
There is a braid-reversing equivalence
$$\mathcal{T}: \Cc\aslnm^0\stackrel{\sim}{\rightarrow}\Cc\aslmn^{0}.$$ 
Furthermore, the equivalence $\Tt$ sends an object $\widehat{a}$ ($a\in C_{n,m}$) to the {\em dual} 
$\widehat{a^{\tau}}^*$ of $\widehat{a^{\tau}}$.% See Sections 2.1 and \ref{tcpr} for definitions.
\end{thm}

\begin{rem} The highest weight of the dual $\widehat{a}^*\in \Cc\aslmn$ is $a^*:=(a_0,a_{m-1},\ldots,a_1)$ if $a=(a_0,a_1,\ldots,a_{m-1})$. Notice now that Theorem \ref{equi} implies $\tau$ commutes with $a\mapsto a^*$.
\end{rem} 

\begin{proof} The proof below is a specialization of the proof of a more 
general result \cite[Theorem 3.6]{KO3}.

Recall that the category $\Cc\aslnm^0 \boxtimes \Cc\aslmn^0$ contains an \'etale algebra $A=\aLam_0$.
Consider the fusion category $\Cc^0:=(\Cc\aslnm^0 \boxtimes \Cc\aslmn^0)_A$. Observe that by
Propositions \ref{ddim} (ii) and \ref{dimac} we have
\begin{equation}\label{eqdi}
\FP \Cc^0=\FP \Cc\aslnm^0=\FP \Cc\aslmn^0.
\end{equation}
We also have the free module functor $\Ff: \Cc\aslnm^0\boxtimes \Cc\aslmn^0\to \Cc^0, \Ff(X)=X\otimes A$. We claim that the restriction of the functor $\Ff$ to 
$\Cc\aslnm^0=\Cc\aslnm^0\boxtimes \be\subset \Cc\aslnm^0\boxtimes \Cc\aslmn^0$
(and, similarly, to $\Cc\aslmn^0$) is fully faithful. Indeed, for $a,b\in C_{n,m}$ we have
$$\Hom_A((\widehat{a}\boxtimes \be)\otimes A,(\widehat{b}\boxtimes \be)\otimes A)=
\Hom(\widehat{a}\boxtimes \be,(\widehat{b}\boxtimes \be)\otimes A)=$$
$$\Hom((\widehat{a}\otimes \widehat{b}^*)\boxtimes \be,A)=
\Hom((\widehat{a}\otimes \widehat{b}^*)\boxtimes \be,\be \boxtimes \be)=\Hom(\widehat{a},\widehat{b}),$$
where the second-last equality follows from Theorem \ref{brules} which says that the only summand of $A$ of the form $-\boxtimes \be$ is $\be \boxtimes \be$. Hence \eqref{eqdi} shows that $\Ff$ restricted to $\Cc\aslnm^0$ (and to $\Cc\aslmn^0$) is an equivalence. Composing the equivalences $\Cc\aslnm^0 \cong \Cc^0 \cong \Cc\aslmn^0$ we get the desired tensor equivalence $\Tt$. %Let $\Zz(\Cc^0)$ be the {\em Drinfeld center} (see \cite[Section 2.3]{KO2}) of $\Cc^0$. The functor $\Ff$ is a {\em central} functor, that is it factors through the forgetful functor $\Zz(\Cc^0)\to \Cc^0$.% in the sense of M\"uger (see for example \cite[Section 2.2]{KO2})

We proceed to show that the equivalence $\Tt$ is braid-reversing. First since $\Ff$ is a central functor, we have two (fully faithful) braided functors $\Ff_{n,m}:\Cc\aslnm^0\to\Zz(\Cc^0)$ and $\Ff_{m,n}:\Cc\aslmn^0\to\Zz(\Cc^0)$, satisfying 
\begin{equation}\label{contain}(\Ff_{m,n}\Cc\aslmn^0)'\supset\Ff_{n,m}\Cc\aslnm^0\end{equation}
by definition of $\boxtimes$. Now we consider $\Cc^0$ as a braided tensor category with the braiding induced by its equivalence with $\Cc\aslmn^0$. This gives us fully faithful braided and braid-reversing functors 
$i:\Cc^0\to\Zz(\Cc^0)$ and $j:\Cc^{0}\to\Zz(\Cc^0)$ respectively, see Section \ref{tcpr} (\ref{intoz}) and (\ref{intoz2}).
%$$i:\Cc^0\to\Zz(\Cc^0),$$
%$$j:\Cc^{0,rev}\to\Zz(\Cc^0),$$
Note that the image of $i$ is the same as the image of $\Ff_{m,n}$. Hence by properties of the Drinfeld center discussed in Section \ref{tcpr}, its centralizer is the image of $j$, that is
$$(\Ff_{m,n}\Cc\aslmn^0)'=j(\Cc^{0}).$$
Since $\FP\Cc^{0}=\FP\Cc\aslnm^0$ and $j$ is fully faithful, the inclusion (\ref{contain}) is an equality and we have
$$\Ff_{n,m}\Cc\aslnm^0=(\Ff_{m,n}\Cc\aslmn^0)'=j(\Cc^{0}).$$
%However we also have that that the image of $\Ff_{n,m}$ is clearly centralized by the image of $\Ff_{m,n}$ by definition of $\boxtimes$, and hence we haveby this way of consideration $\Ff_{m,n}(\Cc\aslmn^0)'$ is equal to the image of $\Cc^{0,rev}$ under the embedding given by the second map in \ref{intoz}. Hence by considering dimensions we have
%$$\Cc^{0,rev}\cong\Ff_{m,n}(\Cc\aslmn^0)'=\Ff_{n,m}\Cc\aslnm^0.$$
The functor $\Tt$ is isomorphic to a composition, 
$$\Cc\aslnm^0\stackrel{\Ff_{n,m}}{\rightarrow}\Ff_{n,m}\Cc\aslnm^0\stackrel{j^{-1}}{\rightarrow}\Cc^{0}\stackrel{i}{\rightarrow}\Ff_{m,n}(\Cc\aslmn^0)\stackrel{\Ff_{m,n}^{-1}}{\rightarrow}\Cc\aslmn^{0},$$
where all functors present except for $j^{-1}$ are braided with $j^{-1}$ braid-reversing. 
Hence $\Tt$ is braid-reversing.

%This gives us the following diagram of braided categories
%$$\Cc\aslnm^0\stackrel{\Ff_{n,m}}{\rightarrow}\Zz(\Cc^0)\stackrel{i}{\leftarrow}\Cc^{0,rev}\stackrel{j}{\rightarrow}\Zz(\Cc^0)\stackrel{\Ff_{m,n}}{\leftarrow}\Cc\aslmn^{0},$$
%where the images of coincident arrows agree (note $\Ff^{rev}$ is the map induced by $j$ under the equivalence of $\Cc^0$ and $\Cc\aslmn^0$). Hence we have a braided tensor functor
%$$\Cc\aslnm^0\stackrel{\Ff_{n,m}}{\rightarrow}i(\Cc^0)\stackrel{i^{-1}}{\rightarrow}\Cc^{0,rev}\stackrel{j}{\rightarrow}\Ff_{m,n}^{rev}(\Cc\aslmn^{0,rev})\stackrel{\Ff_{m,n}^{rev}^{-1}}{\rightarrow}\Cc\aslmn^{0,rev},$$
%which clearly has the same definition as $\Tt$. %This is known to imply the equivalence $\Tt$ constructed above is braid reversing, see \cite[proof of Theorem 3.6]{KO3}.$\Ff_{m,n}$ is the same as the isomorphism $\Cc\aslnm^0\cong \Cc^0$ composed with the first map in \ref{intoz}

Finally, $\Tt(\widehat{a})\cong\widehat{a^{\tau}}^*$ follows from Theorem \ref{brules} and computing:
$$\Hom_A((\widehat{a}\boxtimes \be)\otimes A,(\be \boxtimes \widehat{b})\otimes A)=
\Hom(\widehat{a}\boxtimes \be,(\be \boxtimes \widehat{b})\otimes A)=$$
$$\Hom(\widehat{a}\boxtimes \widehat{b}^*,A)=\left\{ \begin{array}{c} \C \; \mbox{if}\;  b^*=a^\tau; \\0\; \mbox{otherwise}.\end{array}\right.$$ 
%This completes the proof.
\end{proof}

\begin{rem} \label{dyseq}
One can describe the precise relation between the categories 
$\Cc\aslnm$ and $\Cc\aslmn$ following \cite[Theorem 7.20]{FFRS} (see also \cite[Theorem 3.14]{KO3}). Namely, there is an
\'etale algebra $B\in \Cc\aslmn^{rev} \boxtimes \Cc(\asl_{nm})_1$ such that there is a braided
equivalence $\Tt_1:\Cc\aslnm\cong (\Cc\aslmn^{rev} \boxtimes \Cc(\asl_{nm})_1)_B^{dys}$. 
One shows using Theorem \ref{brules}
that $B\cong \oplus_{i=0}^{m-1}\sigma_n^{\otimes i}\boxtimes \aLam_{ni}$ and that $\Tt_1$ sends
$\alam \in \Cc\aslnm$ to the dual of $(\alamt \boxtimes \aLam_{|\lambda|})\otimes B$.
However we feel that these results are more difficult  to use than Theorem \ref{equi}.
\end{rem}

%(see Section \ref{fus})
\subsection{Mirror extensions} The task of constructing \'etale algebras in braided fusion categories is difficult. One general
approach for categories associated with affine Lie algebras is through the use of conformal embeddings as described in 
Section \ref{etce}. Feng Xu observed in \cite{X2} that one can use level-rank duality in order to construct further examples which do not come from conformal embeddings. One can summarize part of his theory categorically as the following consequence of Theorem \ref{equi}:

\begin{prop} \label{xmir}
Let $A\in \Cc\aslnm^0$ be an \'etale algebra. Then $\Tt(A)\in \Cc\aslmn^0$ also has a structure of \'etale algebra.
\end{prop}

\begin{eg}(\cite[Xu]{X2}) 
There exists a conformal embedding $(\asl_2)_{10}\subset (\widehat{\mathfrak{so}}_5)_1$. The 
corresponding \'etale algebra obtained by restriction of the vacuum module is 
$A=\be \oplus \widehat{a}\in \Cc(\asl_2)_{10}$
where $a=(4,6)\in C_{2,10}$. Thus $A\in \Cc(\asl_2)_{10}^0$. Applying Proposition \ref{xmir}
we get an \'etale algebra $\Tt(A)\in \Cc(\asl_{10})_2^0\subset \Cc(\asl_{10})_2$ of the form
$\be \oplus \widehat{a^{\tau}}$ where $a^{\tau}=(0,0,0,1,0,0,0,1,0,0)\in C_{10,2}$. 
Notice that the algebra
$\Tt(A)$ is not associated with any conformal embedding since the {\em conformal dimension} of
$\widehat{a^\tau}$ equals 2. For other interesting examples see \cite{X2, X3}.
\end{eg}

%\newpage

\ \\

 %\address
%\email
 \end{document}